\documentclass[12pt]{article}    
\usepackage[T2A]{fontenc}
\usepackage[cp1251]{inputenc}
\usepackage[a4paper,left=15mm,right=15mm, top=15mm,bottom=20mm]{geometry}
\usepackage[english,russian]{babel}
\usepackage{graphicx}
\usepackage{placeins}
\usepackage{bmpsize}
\usepackage{amsfonts, amsmath, amssymb}
\usepackage{amsthm}
\usepackage{wrapfig}
\usepackage{hyperref}

\graphicspath{}
\DeclareGraphicsExtensions{.pdf,.png,.jpg}

\def\ZZ{{\mathbb Z}}

\newtheorem{thm}{Теорема}
\newtheorem*{thm*}{Теорема}
\newtheorem{lem}{Лемма}
\newtheorem{prop}{Предложение}

\theoremstyle{definition} 

\newtheorem*{defn*}{Определение}

\begin{document}

\title{Модель <<шашки Фейнмана>> с поглощением}
\author{Дмитриев Михаил}
\date{15 Апреля 2022}
\maketitle
\begin{abstract}

 Приводится новое элементарное доказательство теоремы Амбаиниса  и соавторов о том, что в модели ''шашки Фейнмана'' амплитуды вероятности поглощения в начальной точке после 4n шагов пропорциональны числам Каталана. Также впервые вычислены вероятности поглощения в точках, близких к начальной.

\textbf{Ключевые слова:} Шашки Фейнмана, квантовые блуждания, числа Каталана, метод отражений

\textbf{MSC2010:} 82B20, 81T25

\end{abstract}

\section*{Введение}

''Шашки Фейнмана''--- это наиболнее элементарная модель движения электрона. Модель была предложена Ричардом Фейманом в 1965 году \cite{Feynman-Gibbs}. 
Она известна также как одномерное квантовое блуждание \cite{Ambainis}. Модель наиболее примечательна тем, что она наглядно иллюстрирует многие базовые идеи квантовой теории и является одной из архитектур универсального квантового компьютера \cite{Venegas-Andraca-12}. Актуальные обзоры приводятся в \cite{Kempe-09, Skopenkov, Venegas-Andraca-12}.
 
 В основу данной работы лег известный результат Амбаиниса и его коллег \cite{Ambainis} о том, что вероятность возвращения одномерного квантового блуждания в начальную точку равна $\frac{2}{\pi}$ (теорема \ref{tm}). Напомним, что для классического случайного блуждания эта вероятность равна $1$ (теорема Пойа).
  Мы приводим новое элементарное доказательство этого результата и обобщаем его на точки, близкие к начальной (теорема \ref{thm}).

\section*{Определение модели}

Мы начнем с неформального определения модели ''шашки Фейнмана'', а затем дадим точное. Этот раздел почти полностью заимствован из работы \cite{Skopenkov}.
	
\begin{figure}[h]

	\includegraphics[scale = 0.95]{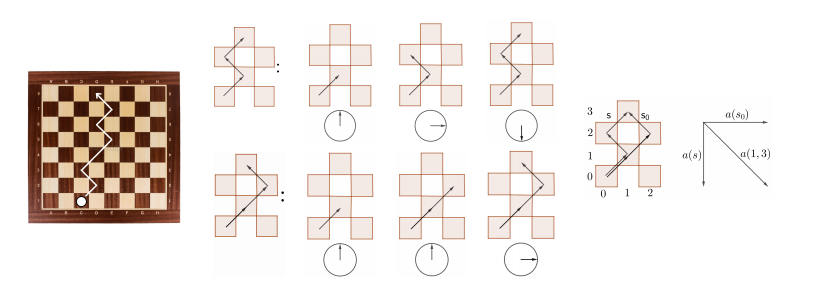}
	\vspace{-.3cm}
	\caption{\cite{Skopenkov} Пути шашек}
	\label{Checker-paths}

\end{figure}	
	
	 На бесконечной шахматной доске шашка ходит на соседнюю по диагонали клетку влево-вверх или вправо-вверх (рис.\ref{Checker-paths}). Каждому пути $s$ шашки сопоставим вектор $a(s)$ на плоскости следующим образом. В начале движения этот вектор направлен вверх и имеет длину $1$.
Пока шашка движется вдоль прямой, вектор не меняется, а после каждого поворота шашки он поворачивается на $90^\circ$ по часовой стрелке (независимо от того, в какую сторону повернула шашка). В конце движения вектор сжимается $2^{(t-1)/2}$ раз, где $t$ --- общее число ходов шашки (т.е. заменяется на вектор такого же направления, но длины $1/2^{(t-1)/2}$). Полученный в итоге вектор и есть $a(s)$. Например, для пути на рис.\ref{Checker-paths} слева вектор $a(s)=(1/8,0)$.

 Обозначим $a(x,t):=\sum_s a(s)$, где суммирование ведется по всем путям $s$ шашки из клетки $(0,0)$ в клетку $(x,t)$, \emph{начинающимся с хода вправо-вверх}. Если таких путей нет, то будем считать $a(x,t):=\vec{0}$. Например, $a(1,3)=(0,-1/2)+(1/2,0)=(1/2,-1/2)$; рис.\ref{Checker-paths} справа.
 
 Число $P(x,t) = |a(x,t)|^2$ называется \emph{вероятностью обнаружения электрона в клетке $(x,t)$, если он испущен из клетки $(0,0)$}.

\begin{figure}
	\centering
	\includegraphics[scale = 0.6]{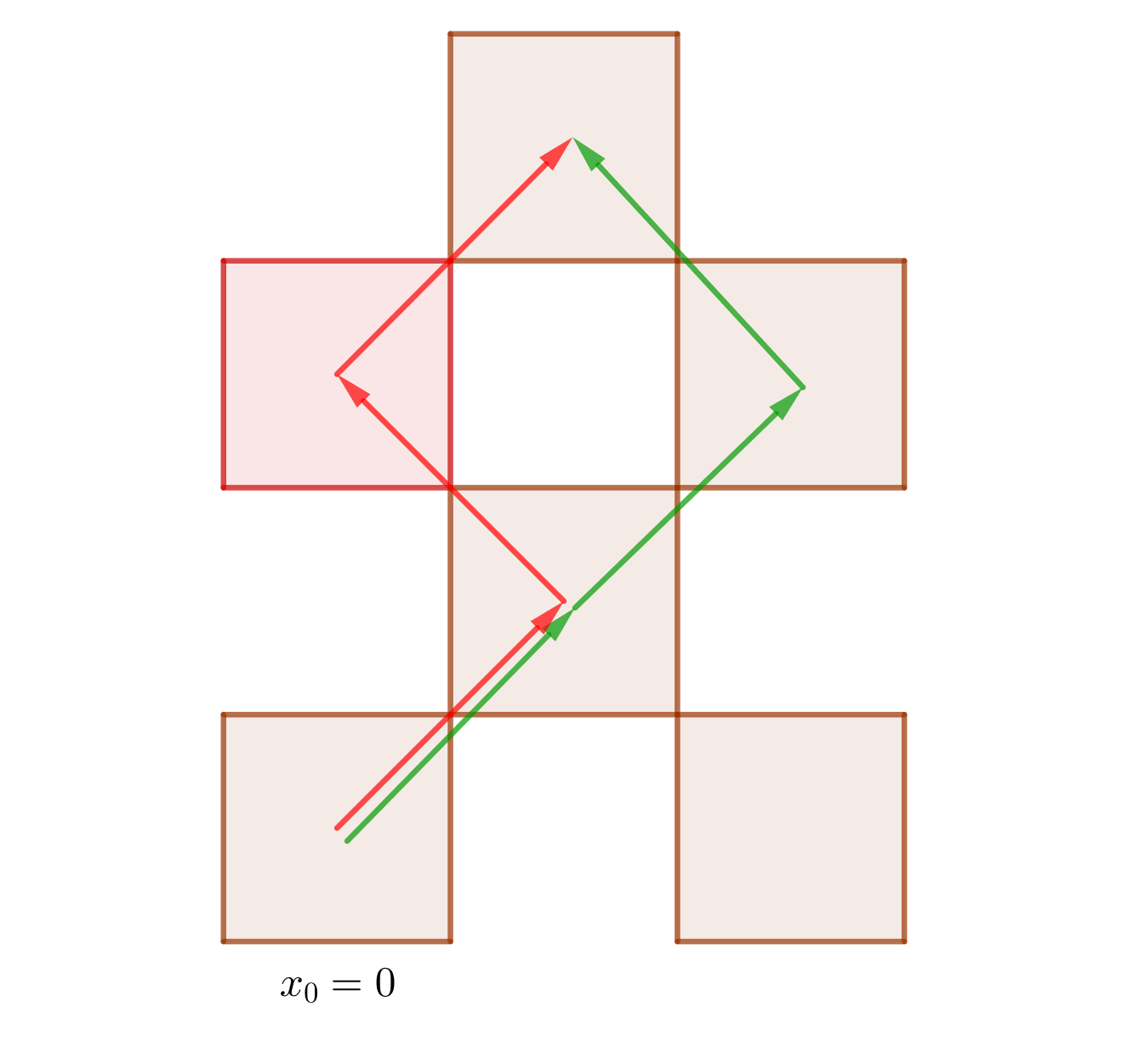}
	\vspace{-.3cm}
	\caption{Пути для модели с поглощением}
	\label{twoways}

\end{figure}

\emph{Вероятность обнаружения электрона в клетке $(x,t)$ при поглощении на прямой $x = x_0$} определяется аналогично вероятности $P(x,t)$, только суммирование производится по путям $s$, не проходящим через клетки вертикали $x=x_0$, за исключением, быть может, начальной и конечной клеток. Обозначим эту вероятность через $P(x,t \,\mathrm{bypass}\, x_0)$.
 Например, на рис. \ref{twoways} красный путь мы не будем учитывать при суммировании, а зеленый будем. Соответственно, для модели с поглощением $a(1,3\,\mathrm{bypass}\, 0)=(1/2,0).$ Значения $a(x,t \,\mathrm{bypass}\, 0)$ при малых $x$ и $t$ показаны на рис. \ref{Count}.

Резюмируем эту конструкцию в виде точного определения.

\begin{defn*} \label{def-basic} \emph{Путь шашки} --- это конечная последовательность целых точек плоскости такая, что вектор из каждой точки (кроме последней) к следующей равен либо $(1,1)$, либо $(-1,1)$. \emph{Поворот} --- это такая точка пути  (не первая и не последняя), что вектор, соединяющий эту точку с предыдущей, ортогонален вектору,  соединяющему её со следующей.
\emph{Стрелка} (или \emph{амплитуда вероятности}) --- это комплексное число
 $$
a(x,t \,\mathrm{bypass}\, x_0):=2^{(1-t)/2}\,i\,\sum_s (-i)^{\mathrm{turns}(s)},
$$ 
где сумма берётся по всем путям $s$ шашки из точки $(0,0)$ в точку $(x,t)$, проходящим через точку $(1,1)$ и не проходящим через точки прямой $x=x_0$, кроме, быть может, начальной и конечной точки пути, а $\mathrm{turns}(s)$ обозначает общее число поворотов в $s$.
Здесь и далее пустая сумма по определению считается равной $0$.
Обозначим  $$P(x, t \,\mathrm{bypass}\, x_0):=|{a}(x, t \,\mathrm{bypass}\, x_0)|^2.
$$
\end{defn*}

\begin{figure}[h]
	\begin{center}
	\includegraphics[scale = 0.95]{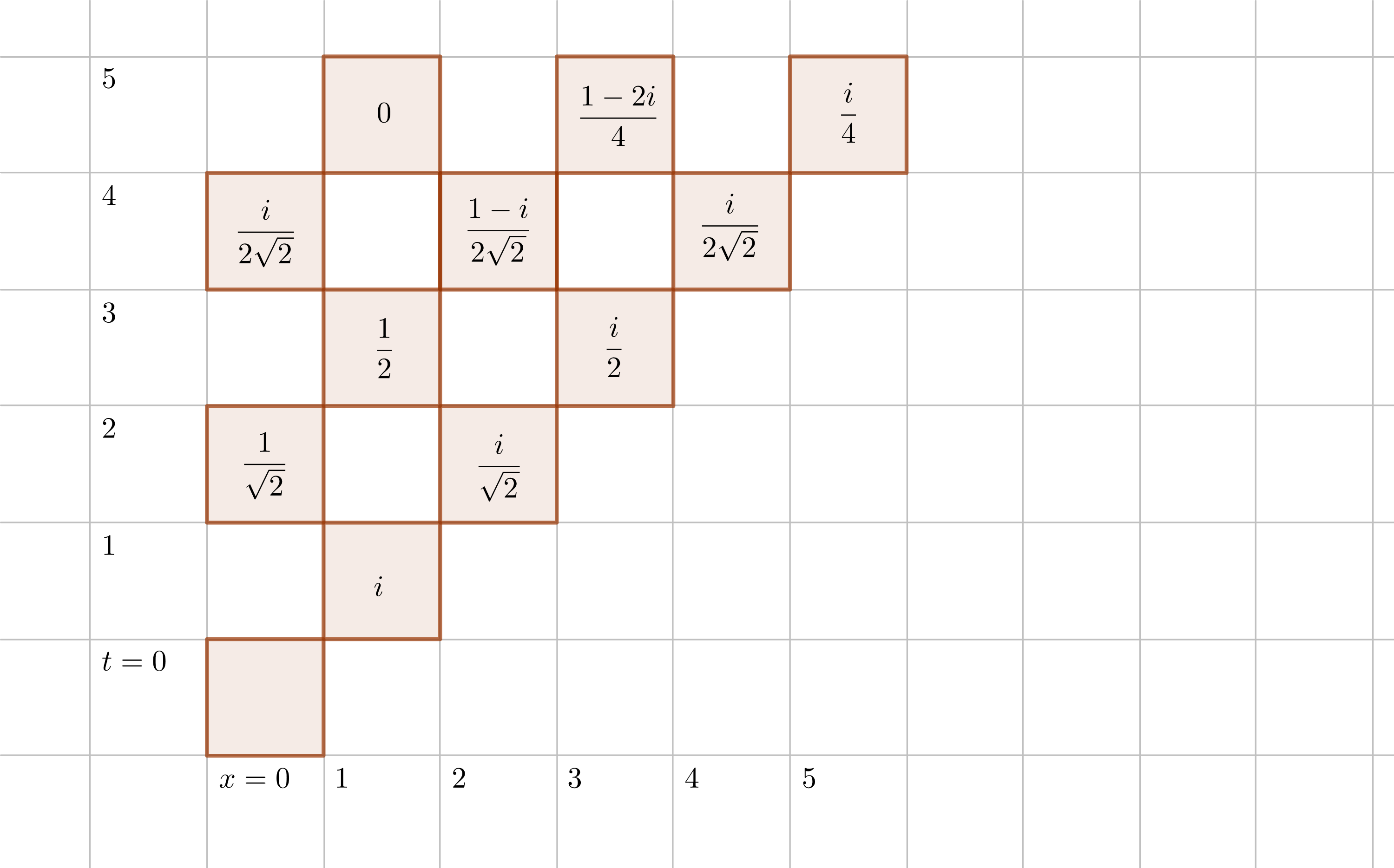}
	\vspace{-.3cm}
	\caption{ Значения $a(x,t \, \mathrm{bypass} \, 0)$ при $0\le x,t\le 5$}
	\label{Count}
	\end{center}
\end{figure}

Исследование данной модели началось со следующего яркого результата.

\begin{thm}[Амбаинис и др., \cite{Ambainis}]
\label{tm} Для любого целого $t > 0$ выполнено
 \begin{equation*}
a(0,t \,\mathrm{bypass}\, 0)= 
 \begin{cases}
   \frac{1}{\sqrt{2}}, &t = 2;\\
    \frac{(-1)^k (^{2k}_k)}{(k+1)2^{2k + 3/2}},& t = 4k + 4, \, \text{где } k \in \ZZ ;\\
   0, &\text{иначе.}
 \end{cases}
\end{equation*} Более того, $$\sum_{t = 1}^\infty P(0,t \,\mathrm{bypass}\, 0) = \frac{2}{\pi}.$$

\end{thm}

Последнее выражение -- это вероятность поглощения электрона в начальной точке.

Далее в разделе ''Вариации'' мы докажем первую часть теоремы способом отличным от такового в статье Амбаиниса.

\section*{Основной результат}

Теперь сформулируем основную теорему. Это утверждение было предложено в качестве гипотезы Г. Минаевым и И. Русских в 2019.

\begin{thm}
\label{thm}
$\sum_{t = 1}^\infty P(3,t \,\mathrm{bypass}\, 3) = 2\sum_{t = 1}^\infty P(-1,t \,\mathrm{bypass}\, {-1}) = \frac{8}{\pi} - 2.$ 
\end{thm}

Основной интерес представляет первое равенство, а второе легко следует из теоремы \ref{tm} и следующей леммы.

\begin{figure}[ht]
	\centering
	\includegraphics[scale = 0.9]{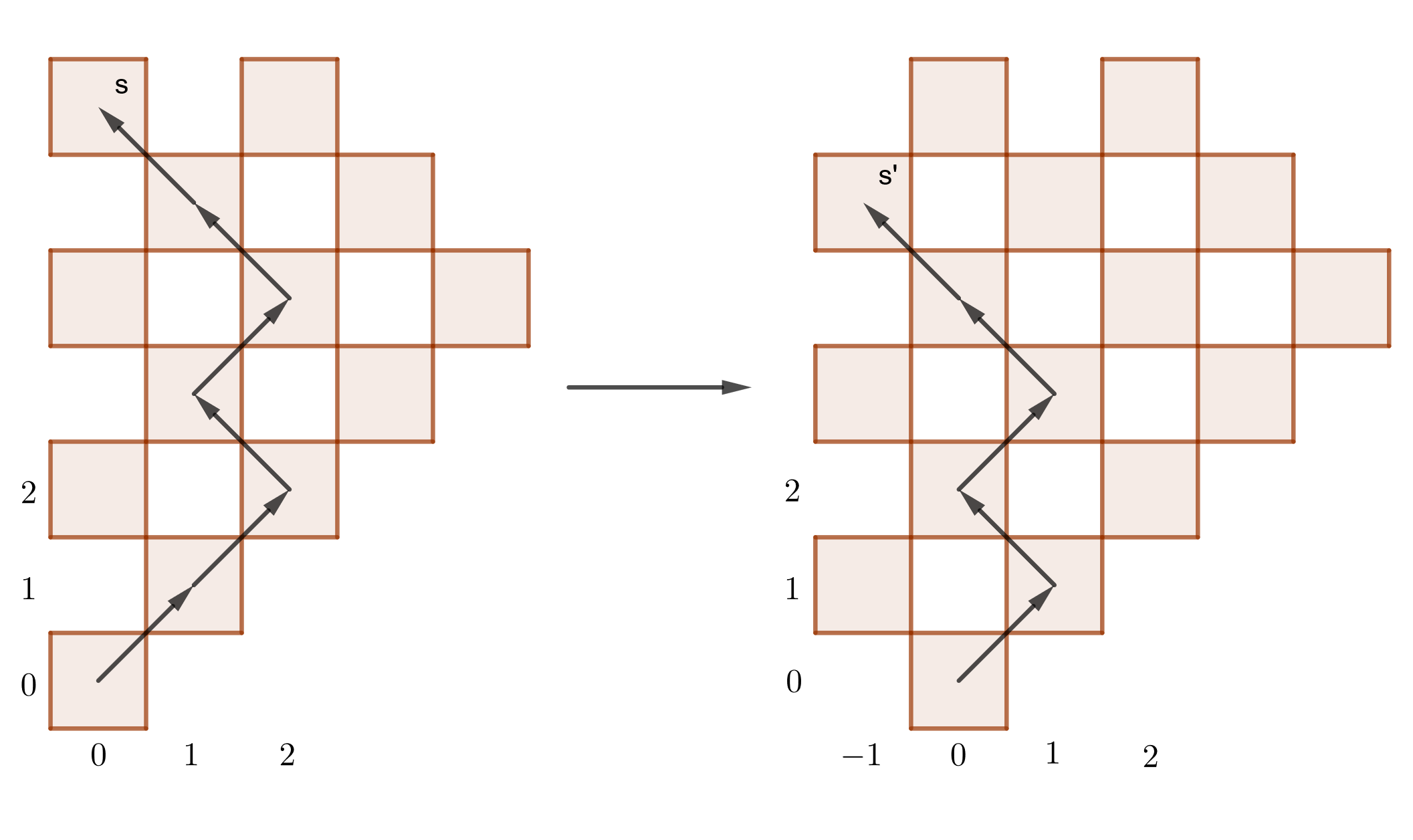}
	\vspace{-.3cm}
	\caption{Отображение путей (см. лемму \ref{l1})}
	\label{pic1}

\end{figure} 

 \begin{lem}
 \label{l1} При $t>2$ верно равенство $ a(-1,t \,\mathrm{bypass}\, {-1}) = \sqrt{2}\,a(0,t+1 \,\mathrm{bypass}\, 0) $.
 \end{lem}

\begin{proof}

 Построим биекцию между путями, дающими вклад в $a(0,t+1 \,\mathrm{bypass}\, 0)$, и путями, дающими вклад в $a(-1,t\,\mathrm{bypass}\,{-1})$. Рассмотрим путь $s$, дающий вклад в $a(0,t+1\,\mathrm{bypass}\,0)$, и построим по  нему путь $s'$, дающий вклад в $a(-1,t \,\mathrm{bypass}\, {-1})$, сдвинув $s$ на вектор $(-1,-1)$ и удалив первый ход (как на рис. \ref{pic1}). Заметим, что для $t>2$ построенное отображение $s \mapsto s'$ --- биекция, так как если бы построенный путь $s'$ начинался с хода влево-вверх, то $s$ проходил бы через точку $(0,2)$ и не мог бы давать вклад в $a(0,t+1\,\mathrm{bypass}\,0)$. Из того, что путь $s'$ короче на один шаг, следует равенство $\sqrt{2}a(s) = a(s')$. Суммируя по всем путям, получаем утверждение леммы.
 \end{proof}

 \begin{lem}
 \label{l2} При $t>2$ верно равенство $a(2,t \,\mathrm{bypass}\, 2) = -i\, a(0,t \,\mathrm{bypass}\, 0)$.
 \end{lem}

\begin{figure}[ht]
	\centering
	\includegraphics[scale = 0.9]{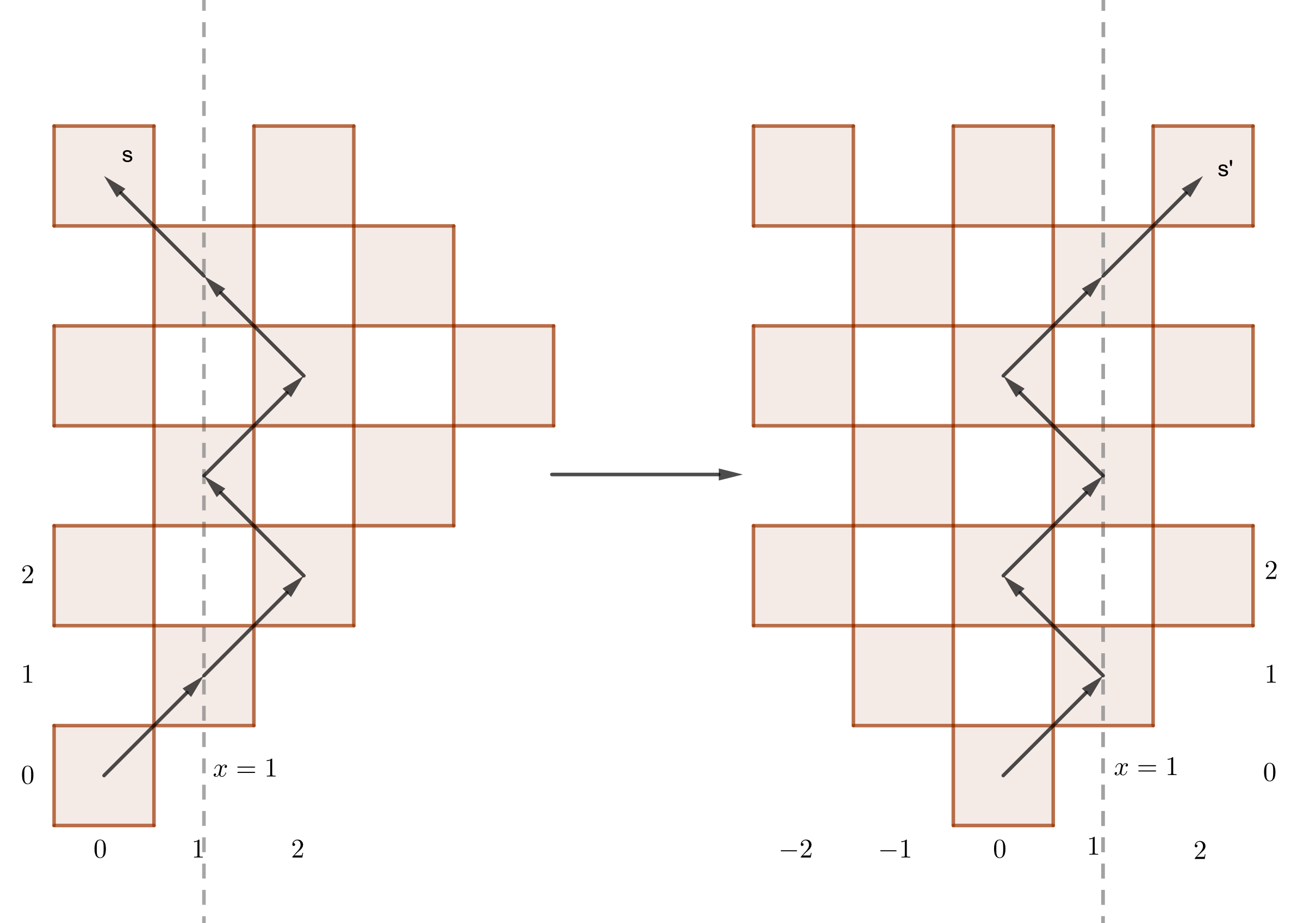}
	\vspace{-.3cm}
	\caption{Отображение путей (см. лемму \ref{l2})}
	\label{pic2}

\end{figure}

\begin{proof}
 Построим биекцию между путями, дающими вклад в $a(0,t \,\mathrm{bypass}\, 0)$, и путями, дающими вклад в $a(2,t \,\mathrm{bypass}\, 2)$. Рассмотрим путь $s$, дающий вклад в $a(0,t \,\mathrm{bypass}\, 0)$, и построим по нему путь $s'$, дающий вклад в $a(2,t \,\mathrm{bypass}\, 2)$, получающийся из $s$ отражением относительно прямой $x=1$, начиная со второго хода (как на рис. \ref{pic2}).
 
  Очевидно, что это биекция. Более того, верно $-i\, a(s) = a(s')$, если $s$ и $s'$ не заканчивались на горизонтали $t=2$. Суммирование по путям $s$ и $s'$ завершает доказательство леммы.
  \end{proof}

\begin{lem}
 \label{l3}
 При $t>3$ верно равенство $$a(3,t\,\mathrm{bypass}\,3) = \sqrt{2} \, a(2,t-1 \,\mathrm{bypass}\, 2) -i\,a(-1,t \,\mathrm{bypass}\, {-1}).$$
\end{lem}

\begin{figure}[ht]
	\centering
	\includegraphics[scale = 0.9]{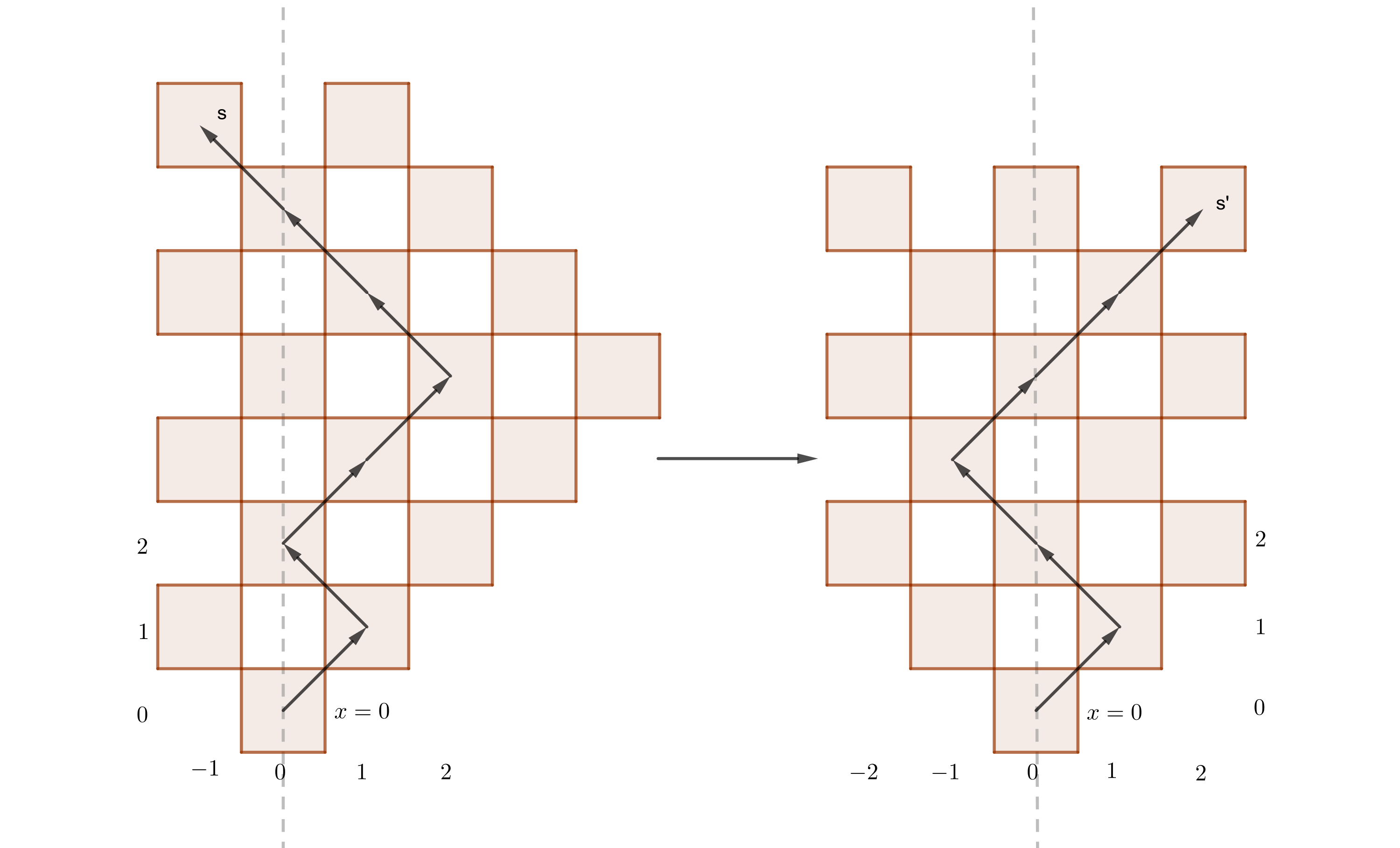}
	\vspace{-.3cm}
	\caption{Отображение путей (см. лемму \ref{l3})}
	\label{pic3}

\end{figure}

\begin{proof}

 Разобьем пути, дающие вклад в $a(3,t \,\mathrm{bypass}\, 3)$, на два множества:

$R$ --- множество путей, которые начинаются с двух ходов вправо-вверх;

$L$ --- множество путей, которые начинаются с хода вправо-вверх, а затем идет ход влево-вверх.

Посторим биекцию множества $R$ и множества путей, дающих вклад в 
$a(2,t-1  \,\mathrm{bypass}\, 2)$, убирая первый ход и сдвигая путь на вектор $(-1,-1)$.

Суммируя по всем таким путям, получаем $$\sum_{s\in R} a(s) = \frac{a(2,t-1 \,\mathrm{bypass}\, 2)}{\sqrt{2}}.$$

Построим биекцию множества $L$ и множества путей, начинающихся двумя ходами вправо-вверх и дающими вклад в $a(-1,t \,\mathrm{bypass}\, -1)$. Для этого отразим путь из $L$, начиная со второго хода, относительно прямой $x=1$ (Это отображение обатно изображенному на рис. \ref{pic2}). Для вычисления $\sum_{s\in L} a(s)$ остается посчитать сумму по путям, которые дают вклад в $a(-1,t \,\mathrm{bypass}\, -1)$ и начинаются с хода вправо-вверх, а затем идет ход влево-вверх. Эти пути мы отобразим в пути, которые дают вклад в $a(2,t-1 \,\mathrm{bypass}\, 2)$, убрав первый ход и сделав сдвиг на вектор $(-1,-1)$, а затем отражение относительно прямой $x=0$ (как на рис. \ref{pic3}). В результате получаем $$\sum_{s\in L} a(s) =  -i \,a(-1,t \,\mathrm{bypass}\, {-1}) + \frac{1}{\sqrt{2}}a(2,t-1 \,\mathrm{bypass}\, 2).$$

Прибавляя вычисленную ранее сумму по $s\in R$, получаем требуемое утверждение.
\end{proof}

\begin{prop}
 \label{prop}
 При $t>3$ верно следующее равенство $$a(3,t\,\mathrm{bypass}\,3)  = -i \, \sqrt{2}\, a(0,t+1 \,\mathrm{bypass}\, 0)  -i \, \sqrt{2}\, a(0,t-1 \,\mathrm{bypass}\, 0).$$

\end{prop}

\begin{proof}  По леммам \ref{l1} -- \ref{l3} имеем $$a(3,t \,\mathrm{bypass}\, 3) =  \sqrt{2}\, a(2,t-1 \,\mathrm{bypass}\, 2) -i\,a(-1,t \,\mathrm{bypass}\, {-1}) = $$
 
  $$ =  -i \,\sqrt{2} \, a(0,t+1 \,\mathrm{bypass}\, 0)  -i \,\sqrt{2} \, a(0,t-1 \,\mathrm{bypass}\, 0).$$ 
  \end{proof}

\begin{proof}[Доказательство теоремы \ref{thm}] Из теоремы \ref{tm} получаем, что $a(0,4t + 4 \,\mathrm{bypass}\, 0) = \frac{(-1)^t (^{2t}_t)}{(t+1)2^{2t + 3/2}}$, а  $a(0,4t+2 \,\mathrm{bypass}\, 0) = 0$ при целых $t>0$. Тогда из предложения \ref{prop} получаем $P(3,4t+1 \,\mathrm{bypass}\, 3) = 2|a(0,4t \,\mathrm{bypass}\, 0)|^2$, а $P(3,4t+3 \,\mathrm{bypass}\, 3) = 2|a(0,4(t+1) \,\mathrm{bypass}\, 0)|^2.$ Так как $\sum_{t = 1}^\infty P(0,t \,\mathrm{bypass}\, 0) = \frac{2}{\pi}$ по теореме \ref{tm}, то $$\sum_{t \in \ZZ}P(3,t \,\mathrm{bypass}\, 3) = \sum_{t > 3} 2|a(0,t+1 \,\mathrm{bypass}\, 0)|^2 + \sum_{t > 3} 2|a(0,t-1 \,\mathrm{bypass}\, 0)|^2 + P(3,3 \,\mathrm{bypass}\, 3) = $$ $$= 2 \left(\sum_{t = 1}^\infty P(0,t \,\mathrm{bypass}\, 0) - P(0,2\,\mathrm{bypass}\, 0) - P(0,4\,\mathrm{bypass}\, 0)\right)+$$ $$ +  2\left( \sum_{t = 1}^\infty P(0,t \,\mathrm{bypass}\, 0) - P(0,2\,\mathrm{bypass}\, 0)\right) + P(3,3 \,\mathrm{bypass}\, 3)=$$ $$ = 2 \left( \frac{2}{\pi} - \frac{1}{2} - \frac{1}{8} \right) + 2\left( \frac{2}{\pi} - \frac{1}{2}\right) + \frac{1}{4} = \frac{8}{\pi} - 2$$
\end{proof}

\section*{Вариации}

Для нового доказательства теоремы \ref{tm} нам потребуется слудующая лемма, аналогичная известному реккурентному соотношению для чисел Каталана.

 \begin{lem}
 \label{polovina}
 Для любого целого $n > 2$ выполнено:
 
 $$a(0,2n \,\mathrm{bypass}\, 0) = \frac{-1}{\sqrt{2}}\sum_{j=2}^{n-2} a(0,2(n-j) \,\mathrm{bypass}\, 0) a(0,2j \,\mathrm{bypass}\, 0). $$
 \end{lem}

\begin{proof}

См. рис. \ref{piclem}. Рассмотрим путь $p$ шашки из точки $(0,0)$ в точку $(0,2n)$, который не имеет других точек пересечения с прямой $x=0$.. Заметим, что при $n > 1$ такой путь пересекает прямую $x=2$ в точке $(2,2)$. Так как $p$ заканчивается на прямой $x=0$, то он пересекает прямую $x=1$ в еще хотя бы одной точке  $(1,2j + 1)$, отличной от $(1,1)$. Здесь $j$ может принимать значения от $1$ до $n-2$. Выберем среди этих $j$ наименьшее. Рассмотрим отображение, которое каждому пути $p$ сопоставляет пару путей $(m,l)$, стартующих из $(0,0)$, где путь $l$ повторяет все ходы  пути $p$ со второго и до $(2j+1)$-го, а путь $m$ начинается с хода вправо-вверх, а затем повторяет ходы пути $p$ с $(2j+2)$-го до последнего. 
 
\begin{figure}[ht]
	\centering
	\includegraphics[scale = 0.50]{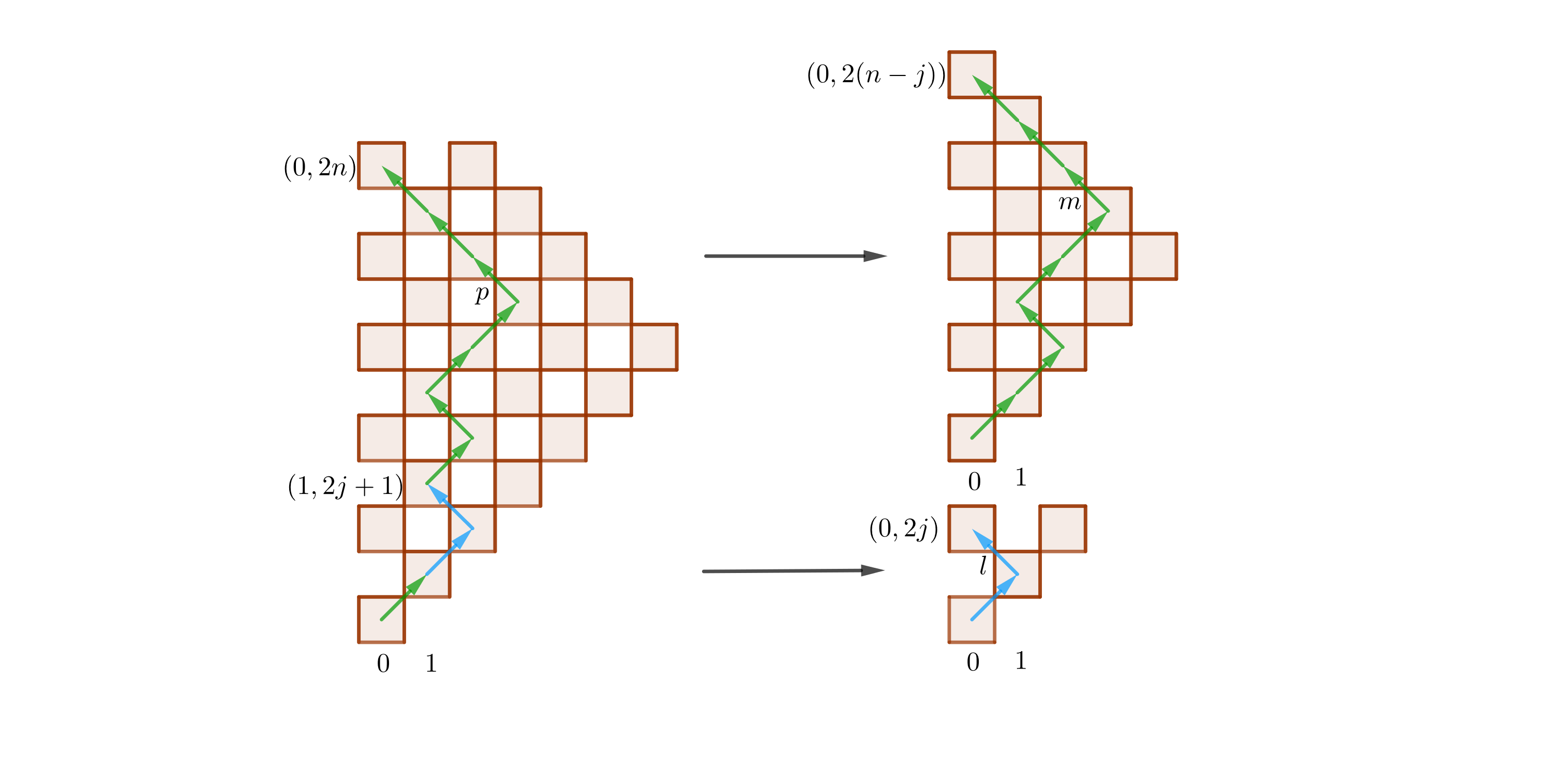}
	\vspace{-.3cm}
	\caption{Отображение путей (см. лемму \ref{polovina})}
	\label{piclem}

\end{figure}

  Покажем, что это отображение --- биекция между множеством путей, дающими вклад в $a(0,2n \,\mathrm{bypass}\, 0)$, и объединением множеств пар путей, дающих вклад в $a(0,2(n-j) \,\mathrm{bypass}\, 0)$ и $a(0,2j \,\mathrm{bypass}\, 0)$ соответственно, по всем $j$ от $1$ до  $n-1$. Для этого построим обратное отображение, т.е. по любой паре $(m,l)$ восстановим путь $p$. Сделаем из точки $(0,0)$ сначала ход вправо-вверх, затем все ходы пути $l$, а затем ходы пути $m$,  кроме первого. Легко видеть, что построенный путь приходит в клетку $(0,2n)$ и он впервые пересекает прямую $x=1$ в точке $(1,2j+1)$. Ясно, что построенное отображение обратно исходному, а значит, они оба --- биекции.

  Тогда при $j <n-1$ имеем $$-\sqrt{2} \, a(p) = -i  (-i)^{\mathrm{turns}(p)}  2^{1 - n} =   i  (-i)^{\mathrm{turns}(m)}  2^{1/2 - n+j}  i  (-i)^{\mathrm{turns}(l)}  2^{1/2 - j} = a(m)a(l), $$ так как путь $p$ имеет такую же длину и на один поворот больше, чем $l$ и $m$ в сумме.
  
   Если же $j = n-1$, то $\sqrt{2} \, a(p)= a(l)a(m),$ так как путь $p$ имеет такую же длину и на один поворот меньше, чем $l$ и $m$ в сумме.
   
    Тогда суммируя по всем путям $p$, получаем:
 
 $$-\sqrt{2}\, a(0,2n \,\mathrm{bypass}\, 0) = \sum_{j=1}^{n-2} a(0,2(n-j) \,\mathrm{bypass}\, 0) a(0,2j \,\mathrm{bypass}\, 0) - $$ $$ - a(0,2(n-1) \,\mathrm{bypass}\, 0) a(0,2 \,\mathrm{bypass}\, 0)  = \sum_{j=2}^{n-2} a(0,2(n-j) \,\mathrm{bypass}\, 0) a(0,2j \,\mathrm{bypass}\, 0).$$
 
 \end{proof}

Теперь мы можем легко доказать теорему \ref{tm}.

\begin{proof}[Доказательство теоремы \ref{tm}]

Достаточно доказать утверждение для четного $t=2n$. Сделаем это по индукции. Базу индукции ($n=1$ и $2$) легко проверить из определений. Докажем переход.

Если $n$ --- четное больше двух, то 

 $$a(0,2n \,\mathrm{bypass}\, 0) = \frac{-1}{\sqrt{2}}\sum_{j=2}^{n-2} a(0,2(n-j) \,\mathrm{bypass}\, 0) a(0,2j \,\mathrm{bypass}\, 0) = $$ 
 $$ = \frac{-1}{\sqrt{2}}\sum_{\substack{j=2 \\ j \text{  четное}}}^{n-2}  \frac{(-1)^{(n-j)/2-1} \binom{n-j-2}{(n-j)/2-1}}{((n-j)/2)2^{n-j - 1/2}} \cdot   \frac{(-1)^{j/2-1} \binom{j-2}{j/2-1}}{(j/2)2^{j - 1/2}}  =  \frac{(-1)^{n/2-1}  \binom{n-2}{n/2-1}}{(n/2)2^{n - 1/2}}. $$

Здесь первое равенство следует из леммы \ref{polovina}, второе --- из предположения индукции, а последнее --- из рекурентной формулы для чисел Каталана.

Если же $n$ --- нечетное больше двух, то

 $$a(0,2n \,\mathrm{bypass}\, 0) = \frac{-2}{\sqrt{2}}\sum_{j=2}^{n-2} a(0,2(n-j) \,\mathrm{bypass}\, 0) a(0,2j \,\mathrm{bypass}\, 0) = $$ 
 $$ = \frac{-2}{\sqrt{2}}\sum_{\substack{j=2 \\ j \text{  четное}}}^{n-2}  0 \cdot   \frac{(-1)^{j/2-1} \binom{j-2}{j/2-1}}{(j/2)2^{j - 1/2}}  =  0. $$

Равенство $\sum_{t = 1}^\infty P(0,t \,\mathrm{bypass}\, 0) = \frac{2}{\pi}$ доказывается так же, как в \cite[доказательство теоремы 8]{Ambainis}.

\end{proof}

\end{document}